\documentclass[letterpaper, 10 pt, conference]{ieeeconf}  % Comment this line out
\IEEEoverridecommandlockouts                              % This command is only
\usepackage[utf8]{inputenc}
\usepackage[T1]{fontenc}
\usepackage{graphicx}
\graphicspath{ {./images/} }

\usepackage{sgame}
\usepackage{amssymb}
\usepackage{amsmath}
\usepackage{xcolor}
\usepackage{float}
\usepackage{comment}
\usepackage{lipsum}
\usepackage{mathbbol}
\usepackage[capitalize]{cleveref}
\usepackage{enumerate}
\usepackage{cite}
\usepackage{mathtools}

\newtheorem{remark}{Remark}
\newtheorem{proposition}{Proposition}
\newtheorem{theorem}{Theorem}

\newtheorem{corollary}{Corollary}
\newtheorem{definition}{Definition}

\newtheorem{lemma}{Lemma}

\DeclareMathOperator{\Equaldef}{\overset{def}{=}}

\title{\LARGE \bf
Rationality and connectivity in stochastic learning for networked coordination games}

\author{Yifei Zhang and Marcos M. Vasconcelos  % <-this % stops a space
\thanks{Y. Zhang and M. M. Vasconcelos are with the Department of Electrical and Computer Engineering, FAMU-FSU College of Engineering, Florida State University,  Tallahassee, FL 32306, USA. E-mails:
        {\tt   yz23r@fsu.edu, m.vasconcelos@fsu.edu}.}%
}

\begin{document}

\maketitle
\thispagestyle{empty}
\pagestyle{empty}

%%%%%%%%%%%%%%%%%%%%%%%%%%%%%%%%%%%%%%%%%%%%%%%%%%%%%%%%%%%%%%%%%%%%%%%%%%%%%%%%
\begin{abstract}

Coordination is a desirable feature in many multi-agent systems such as robotic and socioeconomic networks. We consider a task allocation problem as a binary networked coordination game over an undirected regular graph. Each agent in the graph has bounded rationality, and uses a distributed stochastic learning algorithm to update its action choice conditioned on the actions currently played by its neighbors. After establishing that our framework leads to a potential game, we analyze the regime of bounded rationality, where the agents are allowed to make sub-optimal decisions with some probability. Our analysis shows that there is a relationship between the connectivity of the network, and the rationality of the agents. In particular, we show that in some scenarios, an agent can afford to be less rational and still converge to a near optimal collective strategy, provided that its connectivity degree increases. Such phenomenon is akin to the \textit{wisdom of crowds}. 
\end{abstract}

%%%%%%%%%%%%%%%%%%%%%%%%%%%%%%%%%%%%%%%%%%%%%%%%%%%%%%%%%%%%%%%%%%%%%%%%%%%%%%%%
\section{Introduction}

In strategic decision-making over a social network, every agent must take into account the behavior of its neighbors when choosing an action. Agents in the system are often characterized by a certain level of rationality, which can be modeled on how much weight (if any) is assigned to the actions taken by its neighbors. In one extreme, a completely irrational agent makes decisions by randomly picking an action over a finite discrete set with a uniform distribution. On the opposite extreme, a fully rational agent optimizes its local objective taking into account the actions played by its neighbors, in other words, the agent plays a \textit{best-response}. Somewhere in between these extremes, an agent with bounded rationality will definitely make mistakes, resulting in loss in performance.

In this paper, we investigate how rationality and connectivity interact within a simple learning in games setting \cite{Fudenberg:1998}. In particular, we are interested in establishing whether highly connected agents need to be more or less rational during the process of learning how to play a Nash equilibrium. We show that for a coordination game defined over a network with fixed connectivity degree (regular graph), there is a regime in which higher connectivity implies that the agents can afford to be less rational. This is somewhat related to the \textit{wisdom of the crowds} \cite{Becker:2017} in which an optimal aggregate decision emerges from individual acting independently.

%\newpage

%\begin{figure}[t!]
%    \centering
%\includegraphics[width=0.8\columnwidth]
%{Learning.pdf}
%    \caption{Framework for the log-linear learning process over a graph.}
%    \label{fig:LLL}
%\end{figure}

%\subsection{Related literature}

Coordination games are characterized by the existence of multiple Nash equilibria where the agents play the same action \cite{Easley:2010}. These types of games can be used to model many social-economic phenomena such as technology adoption \cite{Montanari:2010}, political revolutions \cite{Mahdavifar:2017}, task collaboration in robotic networks \cite{Kanakia:2016,Wei:2023,Vasconcelos:2023} and microbiology \cite{Vasconcelos:2022}.  Network games consider the interaction of strategic agents over a graph, and how the graph influences the structure of equilibria \cite{Jackson:2015,Apt:2017,Parise:2021}. A model of networked coordination games subject to cyber-security attacks have been considered in \cite{Paarporn:2021,Paarporn:2021b}. The topic of learning in networked coordination games was investigated in \cite{Arditti:2021,Arditti:2023} for fully rational agents under the best-response dynamics.

The centerpiece of this paper is the analysis of a non-trivial behavior between the rationality of agents and their connectivity in a specific class of coordination games. We look at the convergence of the learning process to a Nash equilibrium when the agents have bounded rationality. This is the case when the agents are not choosing actions which are \textbf{not} the best-response to the actions of their neighbors at any given time. A learning model that captures the rationality of an agent is called \textit{Log-Linear Learning} (LLL) \cite{blume1993statistical,marden2012revisiting}.

LLL is intimately related to the class of potential games \cite{monderer1996potential}. A seminal result shows that when the agents in a potential game use LLL, they converge to one of the Nash equilibria of the game in probability, as the rationality parameter tends to infinity. While in the asymptotic regime, the agents behave predictably, this situation rarely occurs in practice as humans tend to behave with varying degrees of (bounded) rationality. In this regime, the network's connectivity seems to play a very important role. In the bounded rationality regime, a consensus configuration can only be guaranteed with high probability, but never with probability $1$.  
Our main contribution is to show that the minimum level of rationality required to achieve a certain level of coordination may decrease as the connectivity degree of the graph increases. We investigate this issue by constraining our setup to regular graphs \cite[and references therein]{Yaziciouglu:2015}, which greatly reduces the overall complexity of the analysis.

The rest of the paper is organized as follows. In Section II, we introduce our task allocation networked coordination game. In Section III, we show that when the network is described by a regular graph of $N$ vertices and degree $K$, the game is potential. In Section IV, we obtain an alternative expression for the potential function. In Section V, we show that the only maximizers of the potential function are the two consensus configurations and we obtain a complete characterization of when each one is optimal. In Section VI, we obtain the relationship between rationality and connectivity for a large subset of the games considered herein. In Section VII, we present a numerical example. The paper concludes in Section VIII, where we present future research directions.

\section{Problem setup}

Here we define a class of binary networked coordination games that will be the focus of this paper. Let $[N]\Equaldef\{1,2,\ldots, N\}$ denote the set of agents in the network described by an undirected and connected graph $\mathcal{G}\Equaldef([N],\mathcal{E})$.
Two nodes $i,j \in[N]$ are  connected if $(i,j)\in \mathcal{E}$. The set of neighbors of agent $i$ is denoted by $\mathcal{N}_{i} \Equaldef \{j\in [N] \mid (i,j)\in \mathcal{E}\}$. The number of neighbors of agent $i$ is denoted by $|\mathcal{N}_i|$. We assume there are no self loops, i.e., $(i,i) \notin \mathcal{E},$ $i\in[N]$.

Let $(i,j)\in\mathcal{E}$, and suppose that $a_i,a_j\in\{0,1\}$ are the actions played by agents $i$ and $j$, respectively. For $\theta \in \mathbb{R}$, the following bi-matrix game specifies the payoffs for the pairwise interaction between $i$ and $j$.

\begin{figure}[h!]\hspace*{\fill}%
\centering
\begin{game}{2}{2}[$a_i$][$a_j$]
     & $1$ & $0$             \\
 $1$ & $\Big(\frac{1}{|\mathcal{N}_i|}-\frac{\theta}{N},\frac{1}{|\mathcal{N}_j|}-\frac{\theta}{N}\Big)$ & $\Big(-\frac{\theta}{N},0\Big)$ \\
 $0$ & $\Big(0,-\frac{\theta}{N}\Big)$ & $\Big(0,0\Big)$  \\
\end{game}\hspace*{\fill}%
\caption[]{A coordination game with parameter $\theta$ between two players.}
\label{bimatrix}
\end{figure}

%Our framework focuses on the class of binary action graphical coordination games.

%Consider a networked coordination game, where the collection of agents {$[N]=\{1,2,\cdots, N\}$} in the network interacting via an undirected and connected graph $\mathcal{G}=([N],\mathcal{E})$ with definite difficulty. 

%Players are connected to their neighbors in the network, every player has a 

%Actions are binary, that is, $a_i\in\{0,1\}$, $i\in[N]$.

\vspace{5pt}

\begin{remark}[Payoff interpretation]
The payoff structure of the bi-matrix game in \cref{bimatrix} corresponds to a coordination game between two agents. Notice the payoff matrix depends on a parameter $\theta \in \mathbb{R}$, which we refer to as a \textit{task difficulty}. The difficulty is amortized over the total number of agents in the system, $N$. We refer to them as subtasks of difficulty $\theta/N$. However, when agent $i$ decides to take on a subtask, it contributes $1/|\mathcal{N}_i|$ units of effort towards the subtask of each of its neighbors. Therefore, depending on $\theta$, a single agent may not be able to carry out its subtask by itself. This payoff structure simultaneously captures coordination and the graph structure of a networked system in a distributed task allocation problem.
\end{remark}

A binary coordination game between two players is characterized by the existence of two pure strategy Nash equilibria. The following result establishes the range of values of $\theta$ for which the game in \cref{bimatrix} corresponds to a coordination game.

%Based on the concept of the coordination games and the Nash equilibrium, under the following situations, the bi-matrix game between two players can be different types of games:
\vspace{5pt}

\begin{proposition}\label{pure_strategy_equilibria}
    Consider the bimatrix game in \cref{bimatrix}, and $\mathcal{S}_{ij}$ denote its set of pure-strategy Nash equilibria. Let $M_{ij} \Equaldef N/\max\{|\mathcal{N}_{i}|,|\mathcal{N}_{j}|\}$. The following holds:
\begin{equation}
\mathcal{S}_{ij} = 
\begin{cases}
\big\{ (0,0) \big\} & \text{if} \ \  \theta > M_{ij} \\
\big\{ (0,0),(1,1) \big\} & \text{if} \ \ 0 \leq \theta \leq M_{ij} \\
\big\{ (1,1) \big\} & \text{if} \ \ \theta \leq 0.
\end{cases}
\end{equation}
\end{proposition}

\vspace{5pt}

\begin{proof}
    The proof can be obtained by inspection using the definition of a Nash equilibrium \cite{Fudenberg:1998}.
\end{proof}

\vspace{5pt}

\subsection{Coordination games over networks}

We study a network coordination game with $N$ agents, where agent $i$ plays the same action with all of its neighbors $j\in\mathcal{N}_i$. Let $V_i:\{0,1\}^2\rightarrow \mathbb{R}$ be defined as:
\begin{equation}
V_i(a_i,a_j) \Equaldef a_i\Big(
\frac{a_j}{|\mathcal{N}_i|}-\frac{\theta}{N}\Big).
\end{equation}

In a network game, the payoff that one player receives is the sum of all the payoffs of the bi-matrix games $V_{i}(a_{i}, a_{j})$  played with each one of its neighbors. Therefore for the $i$-th player, the utility is determined as follows:
\begin{equation}
    U_{i}(a_{i},a_{-i})\Equaldef \sum\limits_{j\in \mathcal{N}_{i}} V_{i}(a_{i},a_{j}).
\end{equation}
Therefore, the payoff of the $i$-th agent in our game is
\begin{equation}
U_{i}(a_{i},a_{-i}) = a_i\Big(\frac{1}{|\mathcal{N}_i|}\sum_{j\in \mathcal{N}_i}a_j-\frac{|\mathcal{N}_i|}{N}\theta\Big).
\end{equation}

\section{A potential networked  coordination game}

We are interested in obtaining a potential networked coordination game. Potential games \cite{monderer1996potential} have many interesting properties. For example, a potential game implies the convergence of learning algorithms, as will be the case of log-linear learning, which will be the focus of the next section. In this section, we will show that network regularity implies the game is potential. We start with the definition of a potential game. 

\subsection{Potential Games}

\begin{definition}\label{exact_potential}
Let $\mathcal{A}_i$ denote the action set of the $i$-th agent in a game with payoff functions $U_i(a_i,a_{-i})$, $i\in[N]$. Let $\mathcal{A} = \mathcal{A}_1\times \cdots \times \mathcal{A}_n$. A game is an exact potential game if there is a so-called \textit{potential function} {$\Phi$}: 
 {$\mathcal{A}\rightarrow \mathbb{R}$} such that 
\begin{equation}\label{eq:exact_potential}
    U_{i}(a'_{i},a_{-i})- U_{i}(a''_{i},a_{-i}) = \Phi(a'_{i},a_{-i})- \Phi(a''_{i},a_{-i}),
\end{equation}
for all  $a_{i}',a_{i}''\in A_{i}$, $a_{-i}\in A_{-i}$,  $i \in [N]$.
\end{definition}

\vspace{5pt}

\begin{theorem}\label{theorem_EP}
Consider the networked
 coordination game indexed by the parameter $\theta$. Let $K\in\{2,\ldots, N-1\}$. If the network is regular, i.e., if $|\mathcal{N}_{i}|= K, \ i\in[N]$, the game is potential for any value of {$\theta$}. In particular, a potential function for this game is given by {$\Phi(a)$} defined as
\begin{equation}\label{potential_f}
    \Phi(a) \Equaldef \frac{1}{2}\sum\limits_{i\in [N]} \sum\limits_{j\in \mathcal{N}_{i}} \phi(a_{i},a_{j}),
\end{equation}
where 
\begin{equation}\label{eq:potential_two_agents}
\phi (a_i,a_j) \Equaldef \frac{a_ia_j}{K} + (1-a_i-a_j)\frac{\theta}{N}.
\end{equation}
\end{theorem}

\vspace{5pt}

\begin{proof}
For a regular graph of connectivity $K$, the coordination game is given by:
\begin{figure}[h!]\hspace*{\fill}%
\centering\label{fig:bimatrix}
\begin{game}{2}{2}[$a_i$][$a_j$]
     & $1$ & $0$             \\
 $1$ & $\Big(\frac{1}{K}-\frac{\theta}{N},\frac{1}{K}-\frac{\theta}{N}\Big)$ & $\Big(-\frac{\theta}{N},0\Big)$ \\
 $0$ & $\Big(0,-\frac{\theta}{N}\Big)$ & $\Big(0,0\Big)$  \\
\end{game}\hspace*{\fill}%
\caption{A coordination game with parameter $\theta$ between two players on a regular graph of degree $K$.}
\end{figure}

The game above is an exact potential game with potential function determined by the matrix in \cref{tab:exact_potential}. Therefore, the following holds:
\begin{equation}
\phi(a_i',a_j) - \phi(a_i'',a_j) = V_i(a_i',a_j) - V_i(a_i'',a_j),
\end{equation}
for all $a_i',a_i''\in\{0,1\}$ such that $a_i'\neq a_i''$.

\begin{figure}[h!]\hspace*{\fill}%
\centering
\begin{game}{2}{2}[$a_i$][$a_j$]
     & $1$ & $0$             \\
 $1$ & $\frac{1}{K}-\frac{\theta}{N}$ & $0$ \\
 $0$ & $0$ & $\frac{\theta}{N}$  \\
\end{game}\hspace*{\fill}%
\caption{The potential function for the bi-matrix coordination game between two players on a regular graph of degree $K$.}
\label{tab:exact_potential}
\end{figure}

The matrix above corresponds to $\phi$ defined in \cref{eq:potential_two_agents}. Define the function $\Phi: \{0,1\}^N \rightarrow \mathbb{R}$ such that
\begin{equation} \label{eq:potential_candidate}
    \Phi(\mathbf{a}) \Equaldef \frac{1}{2}\sum\limits_{i\in [N]} \sum\limits_{j\in \mathcal{N}_{i}} \phi(a_{i},a_{j}).
\end{equation}
We proceed by verifying that the function in \cref{eq:potential_candidate} satisfies the condition in \cref{eq:exact_potential}.

Let $m \in [N]$, and $a'_m,a''_m \in\{0,1\}$ such that $a'_m\neq a''_m$. Then,
\begin{multline}\label{globalpotental_minus}
     \Phi(a'_{m},a_{-m})-\Phi(a''_{m},a_{-m}) =  \\ \frac{1}{2} \sum_{i\in[N]}\sum_{j\in\mathcal{N}_i} \phi(a_i,a_j)\Bigg|_{(a_m',a_{-m})} \\  - \frac{1}{2} \sum_{i\in[N]}\sum_{j\in\mathcal{N}_i} \phi(a_i,a_j)\Bigg|_{(a_m'',a_{-m})}.
\end{multline}
Then, notice that
\begin{equation}
\sum_{i\in[N]}\sum_{j\in\mathcal{N}_i} \phi(a_i,a_j) = \sum_{j\in\mathcal{N}_m} \phi(a_m,a_j) + \sum_{i\neq m}\sum_{j\in\mathcal{N}_i} \phi(a_i,a_j).
\end{equation}

Recall that 
\begin{equation}
\phi(a_m',a_j) - \phi(a_m'',a_j) = V_m(a_m',a_j) - V_m(a_m'',a_j).
\end{equation}
Therefore,
\begin{multline}\label{eq:intermediate_step_potential}
     \Phi(a'_{m},a_{-m})-\Phi(a''_{m},a_{-m}) =  \\ \frac{1}{2} \sum_{j\in\mathcal{N}_m} \big[ V_m(a_m',a_j) - V_m(a_m',a_j) \big]\\ +
     \frac{1}{2} \sum_{i\neq m} \sum_{j\in\mathcal{N}_i} \phi(a_i,a_j)\Bigg|_{(a_m',a_{-m})} \\ -   \frac{1}{2} \sum_{i\neq m} \sum_{j\in\mathcal{N}_i} \phi(a_i,a_j)\Bigg|_{(a_m'',a_{-m})}.
\end{multline}

The first term in \cref{eq:intermediate_step_potential} is equal to
\begin{equation}\label{eq:half_difference_potential}
\frac{1}{2}\big[ U_m(a'_m,a_{-m}) - U_m(a''_m,a_{-m})\big].
\end{equation}
We proceed with showing that the other two terms also add up to the same value.

For all $i\neq m$ such that $m\notin \mathcal{N}_i$, 
\begin{equation}
\sum_{j\in\mathcal{N}_i} \phi(a_i,a_j)\Bigg|_{(a_m',a_{-m})} \\ = \sum_{j\in\mathcal{N}_i} \phi(a_i,a_j)\Bigg|_{(a_m'',a_{-m})}.
\end{equation}
Define the following set: 
\begin{equation}
\mathcal{S}_m \Equaldef\big\{i \mid  i \neq m \textup{ and } m\in\mathcal{N}_i \big\}
\end{equation}
and evaluate the difference
\begin{multline}
\sum_{i \in \mathcal{S}_m} \sum_{j\in \mathcal{N}_i} \phi(a_i,a_j) \Bigg|_{(a_m',a_{-m})} \\ - \sum_{i \in \mathcal{S}_m} \sum_{j\in \mathcal{N}_i} \phi(a_i,a_j) \Bigg|_{(a_m'',a_{-m})},
\end{multline}
which is equal to 
\begin{multline}\label{eq:difference}
\sum_{i \in \mathcal{S}_m}  \phi(a_i,a'_m) \Bigg|_{(a_m',a_{-m})} \! \! \! - \sum_{i \in \mathcal{S}_m}  \phi(a_i,a_m'') \Bigg|_{(a_m'',a_{-m})}.
\end{multline}
Since 
\begin{equation}
\phi(a_i,a_j) = \phi(a_j,a_i), 
\end{equation}
we have that \cref{eq:difference} is equal to
\begin{equation}
\sum_{i\in\mathcal{S}_m} \phi(a_m',a_i) - \phi(a_m'',a_i).
\end{equation}

Finally, since $\phi$ is a potential function for the two-player game, we have:
\begin{multline}
\sum_{i\in\mathcal{S}_m} V_m(a_m',a_i) - V_m(a_m'',a_i) = \\ U_m(a_m',a_{-m}) - U_m(a_m'',a_{-m}).
\end{multline}
\end{proof}

We have established that this networked coordination game admits an exact potential function when the graph is regular. In the next section, we will obtain a closed form expression for the potential function
for this game, and prove that the only two possible Nash equilibria are $\mathbb{0}_N$ and $\mathbb{1}_N$.

\section{Alternative form for the potential function}

\begin{proposition}
The potential function $\Phi$ for the coordination game over a regular network is given by
\begin{equation}\label{potential_function}
\Phi(a) = \frac{\theta K }{2} - \frac{\theta K}{N} \sum_{i\in[N]} a_i + \frac{1}{2K} a^\mathsf{T}\mathbf{A} a,
\end{equation}
where $\mathbf{A}$ is the adjacency matrix for the undirected regular connected graph that describes the network, $\theta$ is the task difficulty, $K$ is the degree of the regular graph $\mathcal{G}$, $N$ is the number of agents in the network, $a$ is the action profile, and $a_{i}$ is the $i$-th component of $a$.
\end{proposition}

\begin{proof}
The proof is immediate from equations \cref{potential_f,eq:potential_two_agents}. There is only one step that requires some attention, and that is with the following identity:
\begin{equation}\label{eq:identity}
\sum_{i\in [N]} \sum_{j\in\mathcal{N}_i}a_j = K \sum_{i\in[N]}a_i.
\end{equation}
To prove it, notice that the left hand side of \cref{eq:identity} corresponds to the number of active agents in each of the $N$ local neighborhoods of $\mathcal{G}$. Since each agent belongs to exactly $K$ neighborhoods, every active agent is counted $K$ times. Therefore, this sum is equal to $K\cdot\mathbb{1}_N^\mathsf{T}a$. 
\end{proof} 

The seminal result obtained by Monderer and Shapley \cite{monderer1996potential} establishes that the set of Nash equilibria of a potential game coincides with the sets of maximizers of the potential function. Therefore, we procced with the characterization of the solution set of the following optimization problem:
\begin{equation}
\begin{aligned}\label{OriginalProblem}
& \underset{a}{\mathrm{minimize}}
& & \frac{1}{2}K\theta-\frac{K\theta}{N}\sum_{i\in [N]}a_{i} + \frac{1}{2K}a^{\mathsf{T}}\mathbf{A} a \\
& \text{subject to}
& & a \in {\{ 0,1 \}}^{N}.
\end{aligned}
\end{equation}

%The following expression is the system-level potential function for a regular graphical coordination game:
%\begin{equation}
%\Phi^{r}(\mathbf{a}) = \frac{1}{2}k\theta-\frac{k\theta}{N}\sum_{i\in [N]}a_{i} + \frac{1}{2k}(\mathbf{a}^{\mathbf{T}}\cdot L\cdot \mathbf{a}).
%\end{equation}

%In order to maximize the overall potential, we observe the behavior of all agents and formulate the following optimization problem:

Before finding the maximizers of the problem in \cref{OriginalProblem}, we enumerate the following  useful facts:

\begin{enumerate}[(a)]
\item If $\mathbf{A}$ is the adjacency matrix of a \textit{regular} graph of degree $K$, then $\mathbb{1}_N \in \mathbb{R}^N$ is an eigenvector of $\mathbf{A}$, whose corresponding eigenvalue is $\lambda_1 = K $. 

 \item The largest  eigenvalue (in magnitude) of $\mathbf{A}$ is $\lambda_1 = K$. 
 
\item Since $\mathbf{A}$ is real and symmetric, it has $N$ orthogonal eigenvectors $\mathbb{e}_1$, $\mathbb{e}_2$, \dots, $\mathbb{e}_N$. These eigenvectors form an orthogonal basis for $\mathbb{R}^N$.

 \item The graph $\mathcal{G}$ is connected if and only if $\lambda_1 = K$ has multiplicity one.
 \end{enumerate}

\section{Maximizers of the potential function}

\subsection{An equivalent optimization problem}

Consider the orthonormal basis of $\mathbb{R}^N$: $\{ \mathbb{e}_1, \mathbb{e}_2, \dots, \mathbb{e}_N \}$, i.e., the set of eigenvectors of $\mathbf{A}$. Therefore, any action profile $a\in\{0,1\}^{N}$ can be written as a linear combination of $\{\mathbb{e}_i\}_{i=1}^N$. That is,
\begin{equation}
a = \sum_{t=1}^N c_{t}\mathbb{e}_{t}, 
\end{equation}
where $\mathbb{e}_t$ is an eigenvector of $\mathbf{A}$
 such that ${\|\mathbb{e}_t\|}_{2}^2 = 1$,  for $t \in[N]$.
Thus, the following identities hold:
\begin{equation}\label{QF}
\begin{aligned}
    a^{\mathsf{T}} \mathbf{A}  a &= \Big\langle \sum_{t=1}^N c_{t}\mathbb{e}_{t}, \sum_{t=1}^N c_{t}\mathbf{A}\mathbb{e}_{t} \Big\rangle \\
    &= \Big\langle \sum_{t=1}^N c_{t}\mathbb{e}_{t}, \sum_{t=1}^N c_{t}\lambda_{t}\mathbb{e}_{t}  \Big\rangle \\
    &= \sum_{t=1}^{N} \lambda_{t}c_{t}^2, 
\end{aligned}
\end{equation}
where $\lambda_t$ is the $t$-th largest eigenvalue (not necessarily distinct) of $\mathbf{A}$ corresponding to $\mathbb{e}_t$.
Since $a$ is a binary vector, we have that $\sum_{i \in [N]} a_i = {\|a\|}_1 = {\|a\|}_{2}^2$. Then,
\begin{equation}\label{LF}
        \sum_{i \in [N]} a_i = \langle a,a \rangle 
                             = \Big\langle \sum_{t=1}^N c_{t}\mathbb{e}_{t},\sum_{t=1}^N c_{t}\mathbb{e}_{t}\Big\rangle 
                             = \sum_{t=1}^{N} c_{t}^2.
\end{equation}
Using \cref{QF,LF}, we can rewrite the potential function as:
\begin{equation}
\begin{aligned}
     \Phi(a) &= \frac{1}{2}K\theta - \frac{K\theta}{N}\sum_{t=1}^{N} c_{t}^2 + \frac{1}{2K}\sum_{t=1}^{N} \lambda_{t} c_{t}^2 \\ 
    % &= \frac{1}{2}K\theta + \sum_{t=1}^{N} \Big(\frac{1}{2K}\lambda_{t}c_{t}^2 - \frac{K\theta}{N}c_{t}^2\Big)\\
     &= \frac{1}{2}K\theta + \sum_{t=1}^{N} \Big(\frac{\lambda_t}{2K}-\frac{K\theta}{N}\Big)c_{t}^2.
\end{aligned}
\end{equation}
Thus, the optimization problem in \cref{OriginalProblem} can be reformulated in terms of $c_1,c_2,\dots, c_N$:
%\begin{equation}
%\begin{aligned}
%\max_{c_1,c_2,\dots,c_N} \quad & \frac{1}{2}k\theta + \sum_{t=1}^{N} (\frac{\lambda_t}{2k}-\frac{k\theta}{N})c_{t}^2 \\
%\mbox{s.t.} \quad
%&\sum_{t=1}^N c_{t}e_{t} \in {\{ 0,1 \}}^{N}\\
%\end{aligned}
%\end{equation}

\begin{equation}
\begin{aligned}\label{SpectralProblem}
& \underset{c_1,c_2,\dots,c_N}{\mathrm{maximize}}
& & \frac{1}{2}K\theta + \sum_{t=1}^{N} \Big(\frac{\lambda_t}{2K}-\frac{K\theta}{N}\Big)c_{t}^2 \\
& \text{subject to}
& & \sum_{t=1}^N c_{t}\mathbb{e}_{t} \in {\{ 0,1 \}}^{N},
\end{aligned}
\end{equation}
where the optimization variables $c_t \in \mathbb{R},$ $t\in[N]$.

\subsection{A relaxed version of the problem}

The constraint in \cref{SpectralProblem} is difficult to deal with. To find a solution to our optimization problem, we use the following relaxed version of the problem. Consider a new constraint $\sum_{t=1}^N c_{t}^2 \leq N$, which defines a hyper-spherical set that completely covers the feasible set of the original problem in \cref{SpectralProblem}. Therefore, we will solve
\begin{equation}
\begin{aligned}\label{BinaryRelaxation}
& \underset{c_1,c_2,\dots,c_N}{\mathrm{maximize}}
& & \frac{1}{2}K\theta + \sum_{t=1}^{N} \Big(\frac{\lambda_t}{2K}-\frac{K\theta}{N}\Big)c_{t}^2 \\
& \text{subject to}
& & \sum_{t=1}^N c_{t}^2 \leq N.
\end{aligned}
\end{equation}

\begin{theorem}
Let $\mathcal{S}_{\mathcal{G}}$ denote the set of possible Nash equilibria for the network coordination game over a regular connected graph. Then,
\begin{equation}
\mathcal{S}^\star_{\mathcal{G}}= \{\mathbb{0}_N,\mathbb{1}_N \}.
\end{equation}
\end{theorem}

\vspace{5pt}

\begin{proof}
The proof corresponds to showing that the possible solutions to the optimization problem in \cref{SpectralProblem} are $\mathbb{1}_N$ and $\mathbb{0}_N$. Let us solve the relaxed problem in \cref{BinaryRelaxation}. Notice that the objective function is continuous over a compact feasible set. 
Therefore, it always admits a solution. Intuitively, the problem consists of assigning a limited energy budget to each component $c_1,c_2,\dots,c_N$.

Recall that $\lambda_1 > \lambda_2$ for connected graphs. In particular, for a regular graph of degree $K$, $\lambda_1=K$. Therefore, we need to consider three cases:

\begin{itemize}

\item Case 1: Let $\theta$ be such that
\begin{equation}
\frac{\lambda_1}{2K} - \frac{K\theta}{N} < 0 \Rightarrow \frac{\lambda_t}{2K} - \frac{K\theta}{N} < 0, \ \ t\in[N].
\end{equation}
Therefore, 
\begin{equation}
c_t^\star = 0, \ \ t\in[N].
\end{equation}

\item Case 2: Let $\theta$ be such that
\begin{equation}
\frac{\lambda_1}{2K} - \frac{K\theta}{N} > 0.
\end{equation}
Since 
\begin{equation}
\frac{\lambda_1}{2K} - \frac{K\theta}{N} > \frac{\lambda_t}{2K} - \frac{K\theta}{N}, \ \ t\in[N]\backslash\{1\},
\end{equation}
the optimal allocation places the entire energy budget in $c_1$. That is,
\begin{equation}
c_t^\star = \begin{cases}
\pm  \sqrt{N}, & \ \ t=1 \\
0, & \textup{otherwise.}
\end{cases}
\end{equation}

\item Case 3: Let $\theta$ be such that
\begin{equation}
\frac{\lambda_1}{2K} - \frac{K\theta}{N} = 0.
\end{equation}
Then, 
\begin{equation}
\frac{\lambda_t}{2K} - \frac{K\theta}{N} < 0, \ \ t\in[N]\backslash\{1\}.
\end{equation}
Therefore, 
\begin{equation}
c_t^\star = \begin{cases}
 \zeta \in \big[-\sqrt{N},+\sqrt{N}\big] & \ \ t=1 \\
0, & \textup{otherwise.}
\end{cases}
\end{equation}
\end{itemize}

When $c_1,c_2,\dots,c_N$ takes the value $(\pm \sqrt{N},0,\dots,0)^{\mathsf{T}}$ or $\mathbb{0}_N$ in ${\mathbb{R}^N}$, it corresponds to $a= \mathbb{1}_N$ and  $\mathbb{0}_N$, respectively. 
In Case 3, notice that the optimal solution is a continuum. However, only when $\zeta \in\{\pm\sqrt{N},0\},$ we have a corresponding solution in the original feasible set, and these points map to $a= \mathbb{1}_N$ and  $\mathbb{0}_N$, respectively. Since $\mathbb{1}_N$ and $\mathbb{0}_N$ are contained in the feasible set of problem in \cref{OriginalProblem}, the relaxation in \cref{BinaryRelaxation} is exact. Therefore, there is a unique $a^\star \in \{ \mathbb{1}_N, \mathbb{0}_N\}$ that solves  \cref{OriginalProblem} if $\theta \neq \frac{N}{2K}$, with a corresponding optimal values of $(N-K\theta)/2$, and $K\theta/2$, respectively.
\end{proof}

%Here we make three important observations :\\
%(a). Problem \eqref{BinaryRelaxation} is always solvable.  It always preserves a maximizer as $(\sqrt{N},0,0,\dots,0)^{\mathbf{T}}$ or $\mathbf{0} $ in ${\mathbb{R}^N}$, depending on the positivity of $(\frac{\lambda_1}{2k}-\frac{k\theta}{N})$.\\
%(b). The uniqueness of the maximizer can be guaranteed if $\lambda_1$ is strictly greater than any other eigenvalues. This is true since graph $G$ is connected.\\
%(c). When $c_1,c_2,\dots,c_N$ takes the value $(\sqrt{N},0,0,\dots,0)^{\mathbf{T}}$ or $\mathbf{0} $ in ${\mathbb{R}^N}$, they correspond to $\mathbf{a}= \mathbf{1}$ or $\mathbf{a}= \mathbf{0}$. Since vector $\mathbf{1}$ or $\mathbf{0}$ are contained in the feasible set of problem \eqref{OriginalProblem}, this relaxation \eqref{BinaryRelaxation} is exact.\\
%\\
%In conclusion, it is shown that $\mathbf{a}= \mathbf{1}$ or $\mathbf{0}$ is the unique maximizer to problem \eqref{OriginalProblem}. The corresponding potential value is $\frac{1}{2} (N-k\theta)$ or $\frac{1}{2}k\theta$.

\section{Trade-off between rationality and connectivity}

\subsection{Log-Linear Learning}

We assume that all players in the network use a learning algorithm known as \textit{Log-Linear Learning} (LLL)~\cite{blume1993statistical,marden2012revisiting}. LLL is a widely used learning algorithm where each agent updates its action at time $t$ based on its payoff given the actions played by its neighbors at time $t-1$. At each time step {$t>0$}, one agent {$i\in [N]$} is chosen uniformly at random, and allowed to update its current action. All other agents repeat the action taken at the previous time-step. The probability that agent {$i$} chooses action {$a_{i}\in\{0,1\}$} is determined as follows \cite{blume1993statistical}:
\begin{equation}\label{eq:LLL}
    \Pr\big(A_i(t)=a_{i}\big)= \frac{e^{\beta U_{i}\big(a_{i},a_{-i}(t-1)\big)}}{\sum_{a'_i\in \mathcal{A}_{i}}e^{\beta U_{i}\big(a'_{i},a_{-i}(t-1)\big)}}, \  a_i\in \mathcal{A}_i.
\end{equation}

If the game is an \textit{exact potential game}, %\cref{exact_potential}, 
then the Markov chain induced by LLL has a unique stationary distribution $\mu: \{0,1\}^N \rightarrow [0,1]$ given by
\begin{equation}
\mu(a \mid \beta)=
\frac{e^{\beta\Phi(a)}}{\sum_{a'\in \{ 0,1\}^N}e^{\beta \Phi(a')}},
\end{equation}
where $\Phi: \{0,1\}^N \rightarrow \mathbb{R}$ is the potential function.\par
\vspace{5pt}
We are interested in the minimum value of $\beta$ such that all agents coordinate at the optimal action profile of the game with high probability. Consider the following definition, given a relatively small $\delta \in (0,1)$:
\begin{equation}\label{Betamin}
\beta^{\min}(\delta) \Equaldef \min \Big\{ \beta \mid \mu(a^\star \mid \beta) \geq 1-\delta \Big\}.
\end{equation}
where $a^\star$ is the optimal action profile, i.e., the unique maximizer of $\Phi$. \par
\vspace{5pt}

Before discussing the interplay between rationality parameter $\beta$ in LLL and connectivity $K$, we need to address the fact that there may exist multiple regular graphs with the same connectivity degree. These graphs differ from each other up to a similarity transformation on columns (or rows by the fact that the adjacency matrix is symmetric) of their adjacency matrix. Applying the similarity transformation are equivalent to re-assigning indices to agents. Although the single potential value can be different on two isomorphic graphs $\mathcal{G}^1$ and $\mathcal{G}^2$ with the same $K$ and $N$, there exists a unique $\tilde{a}' \in \{0,1\}^N$ such that $\Phi_{\mathcal{G}^1}(a') = \Phi_{\mathcal{G}^2}(\tilde{a}')$. Such $\tilde{a}'$ can be derived by applying the same similarity transformation on $a'$. Therefore, when computing the exact value of $\mu(a'\mid\beta)$ for a specific $a' \neq a^\star$, we should specify and fix a graph $\mathcal{G}$. Nevertheless, $\Phi(a^\star)$ stays the same for any isomorphic graphs with the same degree and so does $\mu(a^\star \mid \beta)$. This is further discussed in the proof of our next theorem.\par
\vspace{5pt}
%Without loss of generality, we can always assume that the adjacency matrix $\mathbf{A}$ is a circulant matrix. i.e. each row of $\mathbf{A}$ can be obtained by performing an one-digit-shift on its previous row. The adjacency matrix of a regular graph can always be made circulant by re-assigning indices to nodes. If $\mathbf{A}$ is a circulant matrix, then the matrix-vector product $\mathbf{A}a$ can be written as the vectorization of a sequence generated by circular convolution of the first row of $\mathbf{A}$ and ${a}$.
%The following lemma is useful when considering the norm of a sequence generated by N-circular convolution:
The following lemma from graph theory is useful when characterizing the existence of a regular graph.

\vspace{5pt}

\begin{lemma}\label{graph}
    A regular graph $\mathcal{G}_K$ with $N$ vertices of degree $K$ exists if and only if $K \in \{2,\dots,N-1\}$ and $NK$ is even.
\end{lemma}

\vspace{5pt}

\begin{corollary}\label{G_K}
    Let $\mathbf{A}_K$ be the adjacency matrix of a regular graph $\mathcal{G}_K$ of degree $K$. The following statements on $\mathcal{G}_{K+1}$ hold: 
    \begin{enumerate}[(a)]
\item Suppose $N$ is even. Then $\mathcal{G}_{K+1}$ always exists. Moreover, the adjacency matrix of the regular graph $\mathcal{G}_{K+1}$ can be formulated in the following sense: there exist two permutation matrices $\mathbf{P}_1$ and $\mathbf{P}_2$ such that, $\mathbf{A}_{K+1} = \mathbf{P}_1(\mathbf{A}_{K}+\mathbf{P}_2)$. 

 \item Suppose $N$ is odd. Then $\mathcal{G}_{K+1}$ does not exist. However, $\mathcal{G}_{K+2}$ exists, and its adjacency matrix is given by $\mathbf{A}_{K+2} = \mathbf{P}_3(\mathbf{A}_{K}+\mathbf{P}_4+\mathbf{P}_5)$ for some permutation matrices $\mathbf{P}_3$, $\mathbf{P}_4$ and $\mathbf{P}_5$.

 \end{enumerate}
 
\end{corollary}
\vspace{5pt}
The next step is to evaluate $a'^{\mathsf{T}}\mathbf{A} a'$ for each $a' \in \{0,1\}^N$. Consider the following bound for the quadratic form $a'^{\mathsf{T}}\mathbf{A} a'$.

\vspace{5pt}

\begin{lemma}\label{Lemma}
%Let $\{\mathbf{x}[n]\}_{n=1}^{N}$ and $\{\mathbf{y}[n]\}_{n=1}^{N}$ be two sequences of length $N$ in $\ell^1$. Let $\| \cdot \|_T$ be a permutation invariant norm on $\ell^1$, that is, let $\mathbf{P} \in \mathbb{R}^{N \times N} $ be a permutation matrix and $x \in \mathbb{R}^N$ be the vectorization of sequence $\{\mathbf{x}[n]\}_{n=1}^{N}$, then $\|\mathbf{P}x\|_T=\|x\|_T$. We have the following inequality:
Let $\mathbf{A}$ be the adjacency matrix of a regular graph of degree $K$. Let $a' \in \mathbb{R}^N$ be a binary vector with $\|a\|_1 = m$. The following inequality holds: 
\begin{equation}\label{Norm}
a'^{\mathsf{T}}\mathbf{A} a' \leq mK , \ \ a' \in \{0,1\}^N
\end{equation}
\end{lemma}

\vspace{5pt}

\begin{proof}
Let $\|\mathbf{A}\|_2$ denote the $\ell^2$ induced operator norm of $\mathbf{A}$, that is, $\|\mathbf{A}\|_2 = \sup_{x \neq \mathbb{0}_N}{\frac{\|\mathbf{A}x\|_2}{\|x\|_2}}$. $\|\mathbf{A}\|_2$ is known to be the largest singular value of $\mathbf{A}$, in our case, $K$. As any operator norm is consistent with the vector norm inducing it, this gives us, for all $a' \in \{0,1\}^N$,
\begin{equation}
    \|\mathbf{A}a'\|_2 \leq \|\mathbf{A}\|_2\|a'\|_2
\end{equation}
Using H\"{o}lder's inequality on $a'^{\mathsf{T}}\mathbf{A} a'$, we obtain:
\begin{equation}
    a'^{\mathsf{T}}\mathbf{A} a' \leq \|a'\|_2  \|\mathbf{A}a'\|_2 \leq \|a'\|_2 \|\mathbf{A}\|_2\|a'\|_2 = mK.
\end{equation}
%Let $x[k]$ and $y[k]$ denote the $k$-th component of sequences $\mathbf{x}$ and $\mathbf{y}$ respectively. For simplicity we define $x[-k] = x[N-k]$. Let $\delta[n]$ be the Discrete Dirac Impulse function with respect to $n$. We have:
%\begin{equation}
%\begin{aligned}
    %\|\mathbf{x} \circledast \mathbf{y}\|_T 
   % &\Equaldef \|\sum_{j=1}^{N}\sum_{k=1}^{N} x[k] y[j-k+1] \delta[n-j]\|_T \\
   % &\leq \sum_{k=1}^{N} \|\sum_{j=1}^{N} x[k] y[j-k+1] \delta[n-j]\|_T \\
    %& = \sum_{k=1}^{N} |x[k]| \|\sum_{j=1}^{N} y[j-k+1] \delta[n-j]\|_T \\
    %& = \|\mathbf{x}\|_1 \|\mathbf{y}\|_T
%\end{aligned}
%\end{equation}
%We need to stress that we are able to switch the order of summations by the fact that the overall result is always finite. 
\end{proof}
%\begin{corollary}\label{pnorms}
    %It's obvious that $\ell^p$ norms are permutation invariant for all $p \in \{1,2,\dots,\infty\}$. We have
    %\begin{equation}\label{Infinity Norm bound}
    %\begin{aligned}
     %   &\|Aa\|_{p} \leq \|a\|_{p}\|x\|_{1}      \\
       % &\|Aa\|_{p} \leq \|a\|_{1}\|x\|_{p}      \\
    %\end{aligned}
   % \end{equation}
    %where $x \in \mathbb{R}^N$ is the first row of circulant adjacency matrix $\mathbf{A}$. Suppose $a$ is a binary vector that contains $m$ 1's and $(N-m)$ 0's, inequality \eqref{Infinity Norm bound} implies that the magnitude of the largest component in vector $(\mathbf{A}a)$ can not exceed $\min\{K,m\}$ by taking $p = \infty$.
%\end{corollary}
%\begin{corollary}\label{l2 Inequality}
%Suppose $a \in \{0,1\}^N$ is an action profile. The the value $a^\mathsf{T}\mathbf{A}a$ is upper bounded by:
%\begin{equation}\label{QuadraticFormBound}
   % a^\mathsf{T}\mathbf{A}a \leq mK
%\end{equation}
   % where m is the number of 1's in $a$. Note that $\|x\|_1$ is exactly the degree of graph $\mathcal{G}$, inequality \eqref{QuadraticFormBound} can be shown by considering $a^\mathsf{T}\mathbf{A}a \leq \|a\|_2\|\mathbf{A}a\|_2$ and applying the first inequality in (\ref{Infinity Norm bound}).
%\end{corollary}
\vspace{5pt}

An interesting interpretation of Lemma \ref{Lemma} is that $a'^{\mathsf{T}}\mathbf{A}a'$ can be seen as an inner product of $a'$ and $\mathbf{A}a'$. This value is obtained by sampling the sequence $\{(\mathbf{A}a')[n]\}_{n=1}^{N}$. The sampling rule is defined by the action profile $a$, where we keep the $i$-th value if $a'_i = 1$ and discard it if $a'_i = 0$. As suggested by the largest eigenvalue of $\mathbf{A}$, the magnitude of the largest component in sequence $\{(\mathbf{A}a')[n]\}_{n=1}^{N}$ can not exceed $K$, and there are at most $m$ such samples, $mK$ becomes a natural upper bound for $a'^{\mathsf{T}}\mathbf{A}a'$. In fact, the insight of this interpretation is captured by the permutation invariance of $\ell^p$ norms and the definition of $\|\mathbf{A}\|_2$.

From this point on, for simplicity, we  ignore the constant term in our potential function $\Phi$, as potential functions of the same coordination game differ from each other by a constant. We use the following expression instead:
\begin{equation}\label{PotentialFunction}
 \hat{\Phi}(a) = -\frac{K\theta}{N}\sum_{i\in [N]}a_{i} + \frac{1}{2K}a^{\mathsf{T}}\mathbf{A} a.
\end{equation}

\begin{theorem}\label{interplay}
Consider a potential game in which the agents use LLL with rationality parameter $\beta\in\mathbb{R}_{\geq 0}$. Suppose all agents are distributed on a regular graph of connectivity $K$. The probability that all agents learn to play the optimal action profile is defined as
\begin{equation}
g(\beta,K)\Equaldef \mu_K(a^\star \mid \beta),
\end{equation} where $a^{\star}$ is a  maximizer of $\Phi$. The function $g$ is strictly increasing in $\beta$. Moreover, $g$ is monotone in $K$ for $K > N/(2\theta)$:
\begin{enumerate}[($i$)]
\item $g(\beta, K) < g(\beta, K+1)$ for even $N$; 
\item $g(\beta, K) < g(\beta, K+2)$ for odd $N$.
\end{enumerate}
\end{theorem}

\vspace{5pt}

\begin{proof}
For monotonicity with respect to $\beta$, computing the derivative of $g$ with respect to $\beta$, we obtain the following equivalence: $\frac{\partial g}{\partial \beta} > 0$ if and only if
\begin{equation}\label{eq:derivative_numerator}
 \hat{\Phi}(a^\star)e^{\beta \hat{\Phi}(a^\star)} \! \! \! \! \! \! \!\sum_{a'\in\{0,1\}^N}e^{\beta\hat{\Phi}(a')} > e^{\beta \hat{\Phi}(a^\star)} \! \! \! \! \! \! \!  \sum_{a'\in\{0,1\}^N} \hat{\Phi}(a')e^{\beta\hat{\Phi}(a')}.
\end{equation}
Since $e^{\beta\hat{\Phi}(a^\star)}>0$, the condition in \cref{eq:derivative_numerator} becomes:
\begin{equation}\label{eq:derivative_numerator2}
\sum_{a'\in\{0,1\}^N} \big(\hat{\Phi}(a^\star) -  \hat{\Phi}(a')\big)e^{\beta\hat{\Phi}(a')} > 0.
\end{equation}
Since $a^\star$ is a maximizer of $\hat{\Phi}$, we have that 
\begin{equation}
\hat{\Phi}(a^\star) \geq \hat{\Phi}(a') , \ \ a' \in\{0,1\}^N.
\end{equation}
Moreover, since there is at least one $\tilde{a} \in \{0,1 \}^N$ such that $\Phi(a^\star) > \Phi(\tilde{a})$, we have that  \cref{eq:derivative_numerator2} holds and consequently, $\frac{\partial g}{\partial \beta}> 0$. The function $g(\beta,K)$ is continuous and strictly increasing in $\beta$, with $g(\beta,K) \rightarrow 1$, as $\beta\rightarrow \infty$.\\

 For monotonicity with respect to $K$, Let $\mathbf{A}_K$ be the adjacency matrix of a fixed regular graph of connectivity $K$. As shown in Corollary \ref{G_K} from Lemma \ref{graph}, we need to consider the parity of $N$.

 Suppose $N$ is even. Then regular graph $\mathcal{G}_{K+1}$ exists and we denote its adjacency matrix by $\mathbf{A}_{K+1}$. For some permutation matrices $\mathbf{P}_1$ and $\mathbf{P}_2$, we have $\mathbf{A}_{K+1} = \mathbf{P}_1(\mathbf{A}_{K}+\mathbf{P}_2)$. Denote $\tilde{a}' = \mathbf{P}_{1}^{-1}a'= \mathbf{P}_{1}^{\mathsf{T}}a'$
 We have that
\begin{equation}\label{Pinvariance}
\begin{aligned}
     a'{^\mathsf{T}}\mathbf{A}_{K+1}a' 
     & = a'{^\mathsf{T}}\mathbf{P}_1(\mathbf{A}_{K}+\mathbf{P}_2)a'\\
      & = a'{^\mathsf{T}}\mathbf{P}_1\mathbf{A}_{K}a' + a'{^\mathsf{T}}\mathbf{P}_1\mathbf{P}_2a'\\
      & = \tilde{a}'{^\mathsf{T}}\mathbf{A}\tilde{a}' + a'{^\mathsf{T}}\mathbf{P}_1\mathbf{P}_2a'\\  
     \end{aligned}
\end{equation}
Note that ($\mathbf{P}_1\mathbf{P}_2$) is again a permutation matrix, its $\ell^p$-induced operator norm $\|\mathbf{P}_1\mathbf{P}_2\|_p$ always equals to 1 since $\ell^p$ norms are permutation invariant. Using the same trick as in Lemma \ref{Lemma}, we have
\begin{equation}\label{Pbound}
\begin{aligned}
     a'^{\mathsf{T}} \mathbf{P}_1\mathbf{P}_2 a' &\leq \|a'\|_{\infty}  \|\mathbf{P}_1\mathbf{P}_2a'\|_1 \\
     &\leq \|a'\|_{\infty} \|\mathbf{P}_1\mathbf{P}_2\|_1\|a'\|_1 = m \\
\end{aligned}
\end{equation}
%\mmv{I think I understand your point now: even when the graphs are circulant, we may have two or more different graphs with the same parameters $N$ and $K$. The question is whether $\beta_{\min}$ is the same for both. I believe it is, but we need to prove it. Due to my brain fog this afternoon, I was missing the point. Now I get it. I will write a brief Matlab code to verify this for two circulant graphs $N=8$ $K=4$.}
Combining (\ref{Pinvariance}) and (\ref{Pbound}) yields that 
\begin{equation}
    a'{^\mathsf{T}}\mathbf{A}_{K+1}a' \leq  \tilde{a}'{^\mathsf{T}}\mathbf{A}_{K}\tilde{a}' + m
\end{equation}
Now let's consider the case $\theta > \frac{N}{2K}$ (the other case $\theta < 0$ will be similar and symmetric, thus it is ommited here). In this case, $a^\star = \mathbb{0}_N $ and $\hat{\Phi}(a^\star) = 0$. For any $a' \neq \mathbb{0}_N$, we have 
\begin{equation}
\begin{aligned}
    &\hat{\Phi}_K(\tilde{a}') = -\frac{K\theta m}{N} + \frac{1}{2K}\tilde{a}'^{\mathsf{T}}\mathbf{A_K} \tilde{a}' \\
    &\hat{\Phi}_{K+1}(a') \leq -\frac{(K+1)\theta m}{N} + \frac{1}{2(K+1)}(\tilde{a}'^{\mathsf{T}}\mathbf{A}_K \tilde{a}'+m)\\
\end{aligned} 
\end{equation}
Taking difference yields that $\hat{\Phi}_K(\tilde{a}')-\hat{\Phi}_{K+1}(a') > 0$ for any $a' \neq \mathbb{0}_N$ since $\theta > \frac{N}{2K}$. Recall that $\tilde{a}' = \mathbf{P}_{1}^{\mathsf{T}}a'$ and $\mathbf{P}_1$ is bijective on $\{0,1\}^N \rightarrow \{0,1\}^N$, we have
\begin{equation}
\begin{aligned}
    \sum_{a'\in \{ 0,1\}^N}e^{\beta \hat{\Phi}_{K}(a')} 
    & = \sum_{\tilde{a}' \in \{ 0,1\}^N}e^{\beta \hat{\Phi}_{K}(\tilde{a}')}\\
    & > \sum_{a'\in \{ 0,1\}^N}e^{\beta \hat{\Phi}_{K+1}(a')}
\end{aligned}
\end{equation}
This shows that $\mu_K(a^\star \mid \beta)$ and $\mu_{K+1}(a^\star \mid \beta)$ have the same numerator as $e^{\beta 0} = 1$, then $\mu_K(a^\star \mid \beta)$ has a larger denominator. Hence $\mu_K(a^\star \mid \beta) < \mu_{K+1}(a^\star \mid \beta)$.\par
Suppose $N$ is odd. Then $\mathcal{G}_{K+1}$ does not exist. The adjacency matrix of regular graph $\mathcal{G}_{K+2}$ is given by $\mathbf{A}_{K+2} = \mathbf{P}_3(\mathbf{A}_{K}+\mathbf{P}_4+\mathbf{P}_5)$. Using the similar inductive procedure as when $N$ is even and define $\hat{a}' = \mathbf{P}_{3}^{-1}a'$, it is evident that
\begin{equation}
    a'{^\mathsf{T}}\mathbf{A}_{K+2}a' \leq  \hat{a}'{^\mathsf{T}}\mathbf{A}_{K}\hat{a}' + 2m
\end{equation}
Consequently  $\hat{\Phi}_K(\hat{a}')-\hat{\Phi}_{K+2}(a') > 0$ for any $a' \neq \mathbb{0}_N $ and $\mu_K(a^\star \mid \beta) < \mu_{K+2}(a^\star \mid \beta)$.
This completes the proof. The function $g$ is strictly increasing in $K$ for $K > \frac{N}{2\theta}$.
\end{proof}
\vspace{5pt}

\begin{remark}
    Note that there exist values of $\theta$ such that we can not observe monotonicity in $K$ at all. This is due to the reason $K \in \{2,\dots,N-1\}$ has a restricted range of values, such that $K$ must stop growing larger before it hit the monotonicity threshold $\frac{N}{2 \theta}$. For  $K \leq \left \lfloor \frac{N}{2\theta} \right \rfloor -1 $, numerical results suggest that the monotonicity is reversed. However, we conjecture that LLL would still show increasing behavior with respect to $K$ in the sense of expected total reward, that is, $\mathbb{E}^{\mu_K} [ \hat{\Phi} ]$ is an increasing function of $K$, as $\hat{\Phi}(a)$ being a random variable defined on $\{0,1\}^N$ with distribution $\mu_K$. We leave this issue to be addressed in future work.
\end{remark}

\vspace{5pt}

Recall the definition of $\beta^{\min}(\delta)$ in \cref{Betamin}. Suppose the failure probability $\delta$ is fixed, then there exists a trade-off between the minimal rationality $\beta^{\min}(\delta)$ and connectivity $K$ since $\mu_K(a^\star \mid \beta)$ is increasing in both $\beta$ and sufficiently large $K$. For $\theta \notin (0,\frac{N}{2})$, a network of better connectivity would allow for a smaller $\beta^{\min}(\delta)$, in order to guarantee that Log-Linear Learning procedure hits the same probabilistic bound. 
%\vspace{5pt}

We proceed to estimate the value of $\beta_K^{\min}(\delta)$. First notice that:
\begin{equation}\label{mu_K}
\mu_K(a^\star \mid \beta) = \frac{e^{\beta\big(-\frac{K\theta}{N}\mathbb{1}_N^{\mathsf{T}}a^\star_{i} + \frac{1}{2K}a^{\star\mathsf{T}}\mathbf{A} a^\star \big)}}{\sum_{a'\in \{ 0,1\}^N}e^{\beta \big(-\frac{K\theta}{N}\mathbb{1}_N^\mathsf{T}a'_{i} + \frac{1}{2K}a'^{\mathsf{T}}\mathbf{A} a' \big) }}.
\end{equation}
Using the inequality from Lemma \ref{Lemma} and $\beta \in \mathbb{R}_{\geq 0}$, we have 
\begin{equation}
    e^{\beta \big(-\frac{K\theta}{N}\mathbb{1}_N^\mathsf{T}a'_{i} + \frac{1}{2K}a'^{\mathsf{T}}\mathbf{A} a' \big) } \leq e^{\beta \big(-\frac{K\theta}{N}m + \frac{mK}{2K} \big) }.
\end{equation}
Counting the number of binary vectors with m 1's in them, we have
\begin{equation}
    \sum_{a'\in \{ 0,1\}^N}\!\!\!\!\!\!\!\!e^{\beta(-\frac{K\theta}{N}\mathbb{1}_N^\mathsf{T}a'_{i} + \frac{1}{2K}a'^{\mathsf{T}}\mathbf{A} a') } \leq \sum_{m=0}^{N} {N \choose m}e^{\beta (-\frac{K\theta}{N}m + \frac{m}{2}) }.
\end{equation}
Multiplying both numerator and denominator in \cref{mu_K} with $e^{\beta K \theta}$, upon using Binomial Theorem, a lower bound on \cref{mu_K} can be obtained as
\begin{equation}\label{LB}
    \mu_K(a^\star \mid \beta) \geq \frac{e^{\beta\hat{\Phi}(a^\star)}e^{\beta K \theta}}{(e^{\frac{1}{2}\beta}+e^{\frac{\beta K \theta}{N}})^N}.
\end{equation}

Note that the right-hand-side of \cref{LB} is also an increasing function of $\beta$ for any $\theta \neq \frac{N}{2k}$. This can be verified by taking its derivative with respect to $\beta$. Moreover, this lower bound matches with the true value of $\mu_K(a^\star \mid \beta)$ when $\beta = 0$ and $\beta \rightarrow \infty$. This means the lower bound in \cref{LB} is asymptotically tight if $\theta \neq \frac{N}{2k}$. Now we are ready to propose our final theorem which quantifies the interplay between $\beta$ and $K$.

\vspace{5pt}

\begin{theorem}
Suppose LLL is performed on a networked coordination game with task difficulty $\theta \neq \frac{N}{2k}$. Then, if 
\begin{equation}\label{Boundforbeta}
\beta \geq {\bigg|\frac{K\theta}{N}-\frac{1}{2}\bigg|}^{-1} \bigg(\frac{\log(1-\delta)}{N}-\log(1-e^{\frac{\log(1-\delta)}{N}})\bigg)
\end{equation}
LLL is guaranteed to force all agents to play the optimal action profile with probability at least $1-\delta$. In other words, $\mu(a^\star \mid \beta) \geq 1-\delta$.

\vspace{5pt}

\end{theorem}
\vspace{5pt}
\begin{proof}
    The proof follows immediately from setting the right-hand-side of \cref{LB} to ($1-\delta$) and solving for $\beta$.
\end{proof}

\section{Numerical results}

\begin{figure}[!htb]
    \centering
\includegraphics[width=0.85\columnwidth]{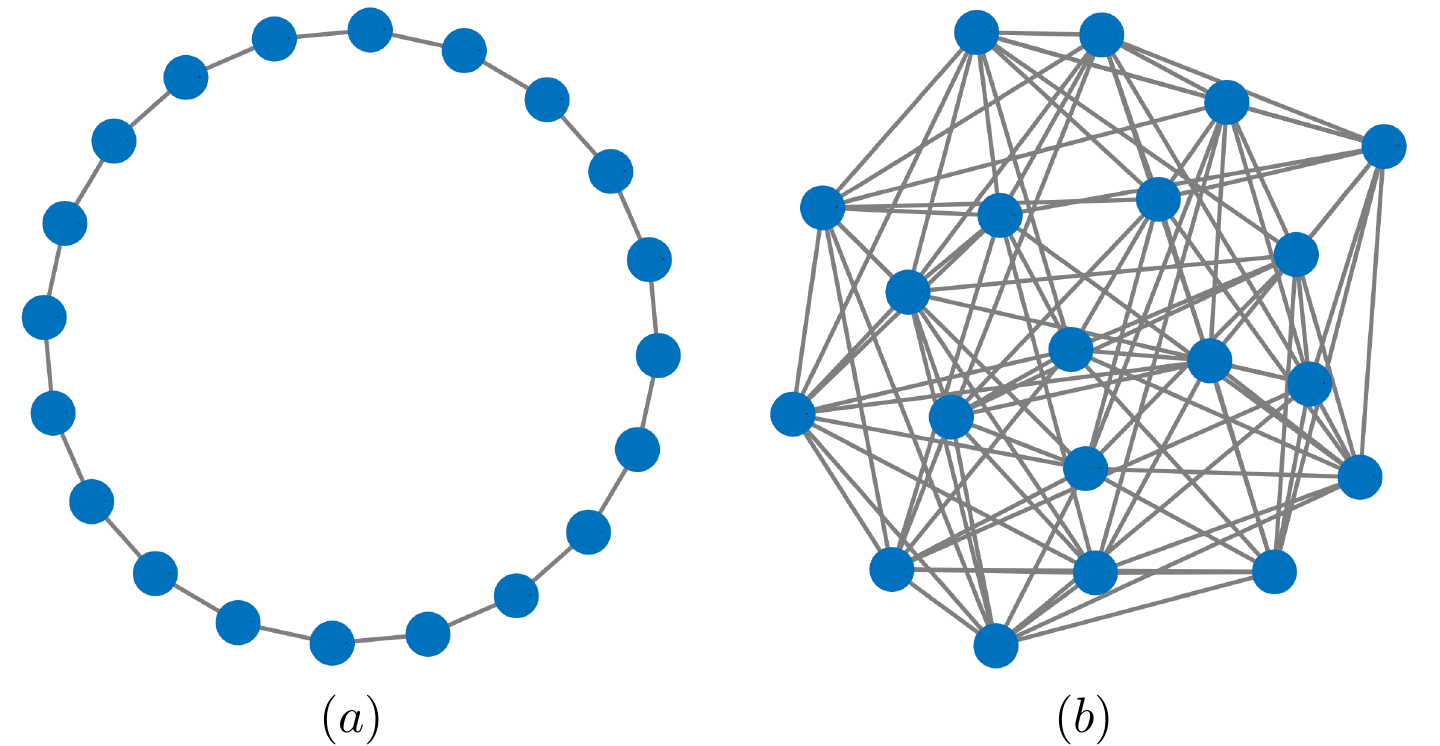}
    \caption{Two examples of connected regular graphs with $N=20$ vertices: $(a)$ $K=2$ and  $(b)$ $K=10$.}
    \label{fig:graphs}
\end{figure}
The above figure shows graphs of the same number of nodes but different connectivity degrees. Simulations are run on these networks and  $\mu_K(a^\star\mid \beta)$ as functions of $\beta$ are plotted for different values of $K$. 

\begin{figure}[!htb]
    \centering
\includegraphics[width=0.85\columnwidth]{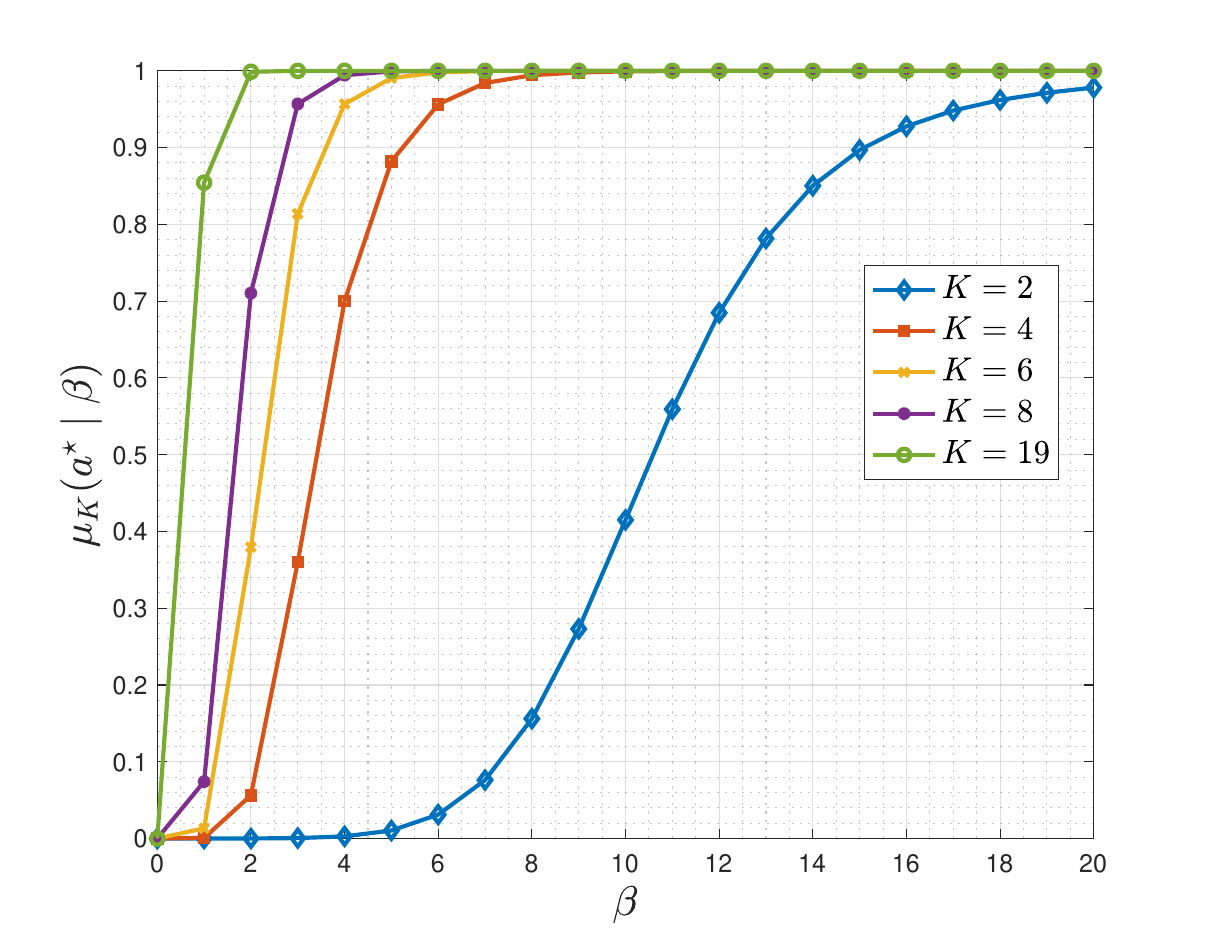}
    \caption{Values of $\mu_K(a^\star\mid \beta)$ for different $K$ when $\theta=5.1$.}
    \label{fig:monotonicity1}
\end{figure}

\begin{figure}[!htb]
    \centering
\includegraphics[width=0.85\columnwidth]{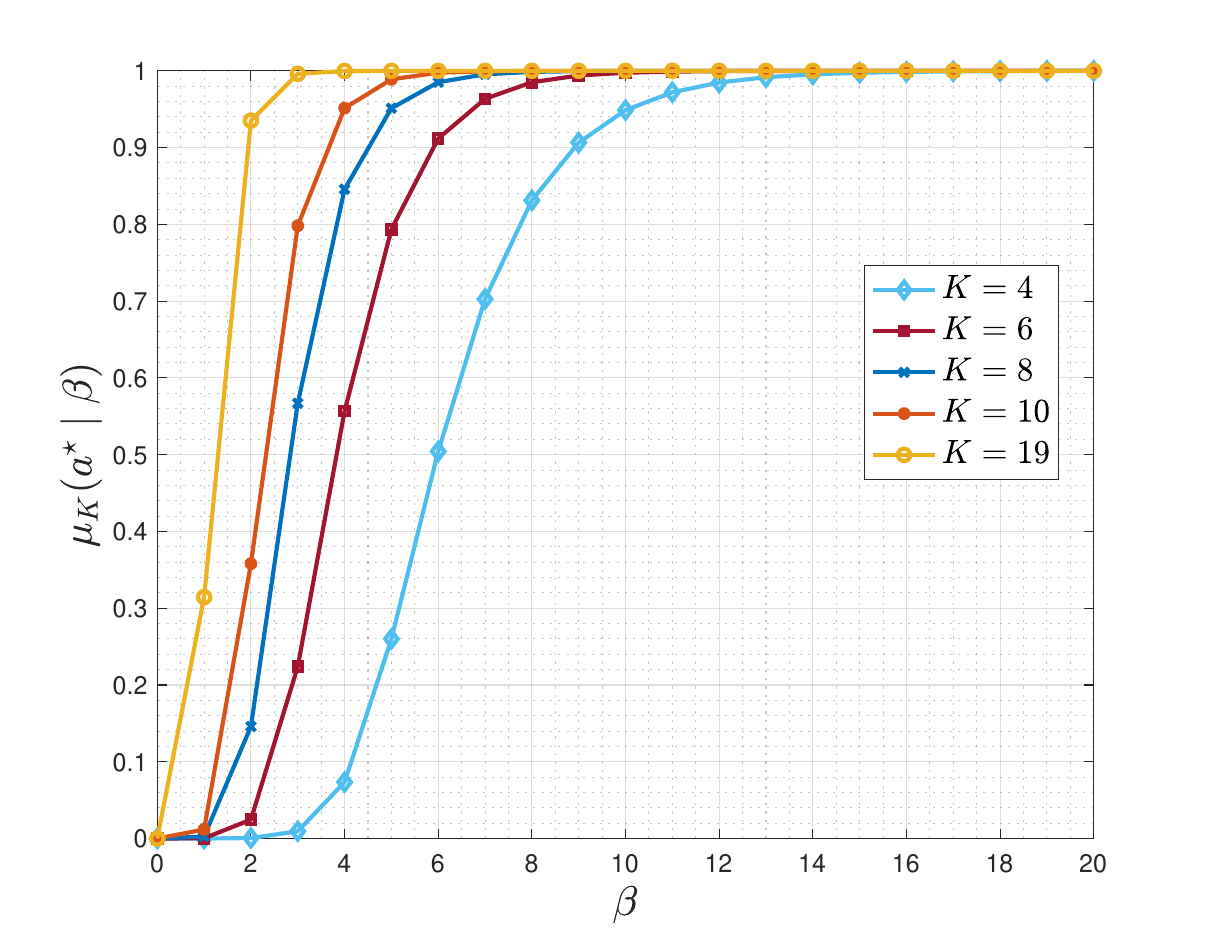}
    \caption{Values of $\mu_K(a^\star\mid \beta)$ for different $K$ when $\theta=3$.}
    \label{fig:monotonicity2}
\end{figure}
We can clearly identify the monotonic behavior of $\mu_K(a^\star\mid \beta)$ in both $\beta$ and $K$ for $\theta = 5.1$ and $\theta = 3$. However, for $\theta = 5.1$, the monotonicity holds for every possible regular graph of degree $K\geq2$, whereas for $\theta = 3$, the monotonicity only holds for $K\geq4$.
%\lipsum[1]

\section{Conclusions and future work}
In this paper we discussed about a coordination game played on a regular network of agents. We showed that the game is potential and gave a closed-form potential function. We proved that the maximizers of the potential function are strictly one of the Nash-equilibria of the original game. We exploited LLL on the network. Upon analysis on the steady state distribution induced by LLL and numerical experiments, we showed that for regular networks with sufficiently large connectivity, there exist a trade-off between connectivity and rationality: better connectivity would allow for a smaller rationality in LLL to achieve the same level of success probability. We also gave an upper bound for the minimal rationality, as a function of connectivity, to guarantee to effectiveness of LLL in the long run. 

%\par

We left the reversed monotonic behavior of steady state probability on poorly connected networks for future work. A hypothesis is presented in the paper: the potential value, which is regarded as a random variable defined on the action space, would show pure monotonicity on connectivity. A rigorous proof for this hypothesis is currently under investigation. We are  also interested in the finite-time analysis, specifically the expected first hitting time of LLL. Stochastic learning in networked coordination games is a rich topic with many open problems, which we are going to explore using analytic and algorithmic approaches.

\bibliography{ref}
    \bibliographystyle{ieeetr}

\end{document}